\documentclass{article}
\usepackage[a4paper]{geometry}
\usepackage{comment}
\usepackage{paralist}
\usepackage{float}
\usepackage{multicol}

\usepackage{nth}
\usepackage{latexsym,amssymb,amsthm,graphicx,float,color}
\usepackage[colorlinks=true, allcolors=blue]{hyperref}
\usepackage{amsthm}
\usepackage[symbol]{footmisc}
\theoremstyle{definition} 
\newtheorem{defn}{Definition} 
\usepackage{graphicx}
\usepackage{braket}

\usepackage{caption}
\usepackage{sidecap}
\usepackage{paralist}
\usepackage{float}

\usepackage[english]{babel}
\usepackage[utf8x]{inputenc}
\usepackage[T1]{fontenc}
\usepackage{listofsymbols}
\usepackage{ica_symbols}

\usepackage[ruled,algo2e]{algorithm2e}

\usepackage{latexsym,amssymb,amsthm,graphicx,float,color, natbib}
\usepackage[colorinlistoftodos]{todonotes}
\usepackage{amsmath}

\usepackage[normalem]{ulem}

\newcommand{\R}{\mathbb{R}}
\newcommand{\inear}{f}

\def\vecc#1{[#1]}

\newcommand{\bp}{\begin{pmatrix}}
\newcommand{\ep}{\end{pmatrix}}

\newcommand\sa[1][]{{#1{s}}_{\rm a}}

\newtheorem{theorem}{Theorem}
\newtheorem{proposition}[theorem]{Proposition}
\newcommand{\Z}{\mathbb{Z}}

\usepackage{pgfplots}
\newif\ifpng

\pngfalse

\title{Focused blind deconvolution}
\author{Pawan Bharadwaj$^*$,
Laurent Demanet,
and
Aim{\'e} Fournier \\ Massachusetts Institute of Technology}

\begin{document}
\maketitle

\begin{abstract}
	We introduce a novel multichannel blind deconvolution (BD)
	method
	that
	extracts \emph{sparse} and \emph{front-loaded} impulse responses from the channel outputs, i.e., 
	their convolutions with a single arbitrary source.
	A crucial feature of this formulation 
	is that it 
	doesn't encode support restrictions on the unknowns, unlike most prior work on BD.
	The indeterminacy inherent to BD, 
	which is difficult to resolve with a traditional 
	$\ell_1$ penalty on the impulse responses,
	is resolved in our method because it seeks
	a first approximation
	where the impulse responses are:
	``maximally white{''} --- encoded as the energy focusing near zero lag of the
		impulse-response auto-correlations; and 
		``maximally front-loaded{''} --- encoded as the energy focusing near zero time 
		of the impulse responses.
	Hence we call the method
	focused
	blind deconvolution
	(FBD).
	The focusing constraints 
	are relaxed as the iterations progress.
	Note that FBD
	requires %
	the  
	duration
	of the channel outputs to be 
	longer than that of the unknown impulse responses.

	A multichannel blind deconvolution 
	problem that is appropriately formulated by sparse and front-loaded impulse responses 
	arises in seismic inversion, where the impulse responses 
	are 
	the Green's function evaluations at different receiver locations, 
	and the operation of a drill bit inputs the 
	noisy and correlated source signature into the 
	subsurface.
	We demonstrate the benefits of FBD 
	using 
	seismic-while-drilling
	numerical 
	experiments, where
	the noisy data recorded at the receivers 
	are hard to interpret,
	but FBD can provide the processing essential to 
	separate
	the drill-bit (source) signature 
	from the interpretable Green's function.

	%

	%
	%

	%
	%
	%
	%
	%
	
	%
	%

\end{abstract}


\section{Introduction}
There are situations 
where seismic experiments are to be performed in environments with a noisy source e.g., when an operating borehole drill is loud enough to reach the receivers.
The {source} generates an {unknown, noisy} signature $s(t)$ at time $t$; {one} typically fail{s} to 
dependably extract the source signature 
despite deploying an attached receiver.
For example, 
the exact signature of the 
operating drill bit 
in a borehole environment cannot be recorded because there would always be some material interceding 
before the receiver \citep{aminzadeh2013geophysics}.
The noisy-source signals propagate through the 
subsurface, and result in the data at the receivers, denoted by $d_i(t)$.
Imaging of the  
data to characterize the subsurface (seismic inversion) 
is only possible when 
they are {deconvolved to discover} 
the subsurface Green's function.
Similarly, 
in room acoustics,
the speech signals $s(t)$ recorded as $d_i(t)$ at a microphone array
are distorted and sound reverberated
due to the reflection of walls, furniture and other objects.
Speech recognition and compression is simpler 
when the reverberated records $d_i(t)$ 
at the microphones are deconvolved 
to recover the 
clean speech signal 
\citep{liu2001blind,yoshioka2012making}.
%

The response of many such physical systems 
to a noisy source is to produce multichannel outputs.
The $\Nr$ observations or channel outputs, in the noiseless case, are 
modeled as the output of a linear system that
convolves (denoted by $\ast$) a 
source (with signature $s(t)$) with
the impulse response function: 
\begin{eqnarray}
	d_{i}(t)=\{s\ast g_i\}(t). 
\end{eqnarray}
Here, $g_i(t)$ is the $i^\text{th}$ channel impulse response and
$d_i(t)$ is the $i^\text{th}$ channel output.
The impulse responses 
contain physically meaningful information about the channels.
Towards the goal of 
extracting the vector of impulse responses 
$\vecc{g_1(t),\ldots,g_{\Nr}(t)}$ or simply $\vecc{g_i}$
and the source function $s(t)$,
we consider 
an unregularized 
least-squares fitting of the
channel-output vector 
$\vecc{d_1(t),\ldots,d_{\Nr}(t)}$ or $\vecc{d_i}$. 
This corresponds to 
the least-squares multichannel 
deconvolution \citep{amari1997multichannel,douglas1997multichannel,sroubek2003multichannel}
of the channel outputs  
with an unknown blurring kernel, i.e., the source signature.
It is well known that
severe non-uniqueness issues {are} inherent to multichannel blind deconvolution (BD);
there could be many
possible estimates of $\vecc{g_i}$,
which when convolved with the corresponding $s$
will result in the recorded $\vecc{d_i}$ (as formulated in eq. \ref{eqn:decont2} below).

Therefore, in this paper, 
we add two additional constraints to the BD framework
that seek a solution where $\vecc{g_i}$ are:
\begin{enumerate}
\item
	\emph{maximally white}
		---
	encoded as the
	energy 
	focusing 
	near zero lag (i.e., energy diminishing at non-zero lags)
	of the impulse-response auto-correlations and 	
\item
	\emph{maximally front-loaded} ---
	encoded as the
	energy 
	focusing near zero time
	of the most front-loaded impulse response.	
\end{enumerate}
We refer to them as \emph{focusing} constraints. 
They are not equivalent to $\ell_1$ minimization,\footnote{That is, minimizing $\sum_t|g_i(t)|$ to promote sparsity.} although they also enforce a form of sparsity.
These are 
relaxed as the iterations progress to enhance the fitting of the channel outputs.
Focused blind deconvolution (FBD) 
employs the 
focusing constraints to resolve the indeterminacy inherent to the BD problem.
We identify that it is more 
favorable to use the constraints
in succession after 
decomposing the BD problem 
into two separate 
least-squares optimization problems.
The first problem, where it is sufficient to employ the first constraint,
fits the \emph{interferometric or cross-correlated} channel outputs \citep{bharadwajibd},
rather than the raw outputs, and solves for the interferometric impulse response.
The second problem relies on the outcome of the previous problem and 
completes FBD by employing the second constraint
and solving for the impulse responses from their cross-correlations.
This is shown in the Figure~\ref{fig:fbd_parts}.
According to
our numerical experiments, 
FBD
can effectively 
retrieve $\vecc{g_i}$
provided the following 
conditions are met:
\begin{itemize}
\item  
	the {duration} length of the unknown impulse responses
	should be much {briefer} than 
	{that of} the channel outputs;
\item  the channels are \emph{sufficiently dissimilar} in the sense of their 
	transfer-function polynomials being \emph{coprime} in the $z$-domain.
\end{itemize}
In the seismic inversion context, 
the first 
condition is {economically beneficial}
, as {usual drilling practice enables us to }
record noisy data for a time period
much longer compared to the wave-propagation time.
Also, since drilling is anyway necessary, its use as a signal source to estimate $\vecc{g_i}$ is a free side benefit.
We show that 
the second condition {can} be 
satisfied in the seismic experiments by deploying \emph{sufficiently dissimilar} receivers, as defined below, 
which may yet be arrayed variously in a borehole, or surface-seismic geometry.

\cite{xu1995least} showed 
that multichannel blind deconvolution is dependent 
on the condition that the transfer functions are coprime, i.e.,
they do not share common roots in the $z$-domain.
The BD algorithms in \cite{subramaniam1996cepstrum,huang2002adaptive} 
are also based on this prerequisite.
In this regard,
due to the difficulty of factoring the high order channel polynomials,  
\cite{gaubitch2005adaptive} proposed a method for identification of common roots of 
two channel polynomials.
Interestingly,
they have observed 
that the roots do not have to be exactly equal to be considered common in BD.
\cite{khong2008algorithms} uses clustering to efficiently extract clusters of near-common roots.
In contrast to these methods,
FBD doesn't need the identification of the common roots of the channel polynomials.

{S}urvey{s} of {BD}
algorithms in the signal and image processing literature {are} given in 
\cite{kundur1996blind} and \cite{campisi2016blind}. 
A series of results on blind deconvolution appeared in the literature 
using different sets of assumptions on the unknowns.
The authors in
\cite{ahmed2015convex} and \cite{li2016rapid}
show that BD can be efficiently solved 
under certain subspace conditions on both the source signature and impulse responses even in the single-channel case.
\cite{ahmed2016leveraging}
showed the recovery of the unknowns in multichannel BD 
assuming that the source is sparse in some known basis
and the impulse responses belong to known random subspaces.
The experimental results in \cite{romberg2013multichannel} show the successful joint recovery 
of Gaussian impulse responses with known support that are 
convolved with a single Gaussian source signature.
BD algorithms with various assumptions on input statistics are proposed in \cite{tong1994blind, tong1995blind} and 
\cite{tong1998multichannel}.
Compared to the work in these articles, FBD doesn't require any 
assumptions on
\begin{inparaenum}
\item support of the unknowns,
\item statistics of the source signature and
\item the underlying physical models;\footnote{Some seismic BD algorithms 
	design deconvolution operators using an estimated subsurface velocity model \citep{haldorsen1995walk}.}
\end{inparaenum}
although, 
it does apply a type of sparsity prior on the $\vecc{g_i}$. 
Note also that 
regularization in the sense of minimal $\ell_1$ i.e., mean-absolute norm,
as some methods employ, does not fully address the type {of} indeterminacy associated with BD.

Deconvolution is 
also an
important step in the processing workflow 
used by exploration geophysicists
to improve
the resolution of the seismic records \citep{ulrych1995wavelet, liu2003survey, van2008robust}.
\cite{robinson1957predictive} developed predictive decomposition \citep{wold1938study}
of the seismic record into a source signature and a white or uncorrelated time sequence corresponding to the Earth's impulse response.
In this context, the impulse responses $\vecc{g_i}$ correspond 
to the 
unique subsurface Green's function {$g(\vec{x},t)$} 
evaluated at the receiver locations {$\vecc{\vec{x}_i}$}, 
where the seismic-source signals are recorded.
Spiking deconvolution \citep{robinson1980geophysical, yilmaz2001seismic} 
estimates a Wiener filter that
increases the \emph{whiteness} of the seismic records, therefore,
removing the effect of the seismic sources.
In order to alleviate the non-uniqueness issues in blind deconvolution,
recent algorithms 
in geophysics{:}
\begin{itemize}
\item take advantage of the multichannel nature of the seismic data \citep{kaaresen1998multichannel,kazemi2014sparse, nose2015fast, liu2016sparse};
\item sensibly choose the initial estimates of the 
$\vecc{g_i}$ in order to converge to a desired solution \citep{liu2016sparse}; {and/or}
\item constrain the sparsity of the 
$\vecc{g_i}$ \citep{kazemi2014sparse}.
\end{itemize}
\cite{kazemi2016surface} used sparse BD 
to estimate source and receiver wavelets while processing seismic records acquired on land.
The {BD} algorithms in the current geophysics literature
handle roughly impulsive source wavelets that are due to 
{well-controlled} sources, as opposed to the {noisy and uncontrollable sources} in FBD{,
about which we assume very little}.
It has to be observed that 
building initial estimates of the 
$\vecc{g_i}$ is difficult for {any algorithm,} 
as the {functional} distance{s} between the 
$\vecc{d_i}$ and the actual 
$\vecc{g_i}$ are quite large. {Unlike standard methods,}
FBD does not require an extrinsic starting guess.

The Green's function retrieval is also the subject of 
\emph{seismic interferometry}
\citep{schuster2004interferometric,snieder2004extracting,shapiro2005high,wapenaar2006unified,curtis2006seismic,schuster2009seismic}, where
the
cross-correlation (denoted by $\otimes$) 
between the records at two receivers with indices $i$ and $j$, 
\begin{eqnarray}
	d_{ij}(t) = \{d_{i}\otimes d_{j}\}(t) = 
	\{
	\sa
	\ast 
	g_{ij}
	\}
	(t),
	\label{eqn:intro1}
\end{eqnarray}
is
treated as a proxy for 
the cross-correlated or \emph{interferometric} Green's function 
$g_{ij}{=g_{i}\otimes g_{j}}$.
A classic result in interferometry states that 
a summation on the $g_{ij}$ 
over various {noisy sources}, 
evenly distributed in space,
will result in the Green's function due to a \emph{virtual source} 
at one of the receivers \citep{wapenaar2006green}.
In the absence of multiple evenly dis\-tributed {noisy sources},
the inter\-fero\-metric Green's functions can {still} be directly used for imaging
\citep{claerbout1968synthesis,draganov2006seismic,borcea2006coherent,demanet2017convex,vidal2014retrieval},
{although this requires knowledge of the source signature.} 
The above equation shows
that the 
goal of interferometry, i.e.,
construction of 
$g_{ij}$
given {$d_{ij}$},
is impeded 
by the source  
auto-correlation $\sa {=s\otimes s}$.
In an impractical
situation with 
a zero-mean white {noisy source},
the {$d_{ij}$} 
would be precisely proportional to {$g_{ij}$}
{; but this is not at all realistic, so we 
don't assume a white source signature in FBD and eschew any concepts like virtual sources.

The {failure} 
of 
seismic {noisy sources} {to be white}\footnote{
For example,
the noise generated 
by
drill bit operations is heavily correlated 
in time \citep{gradl2012analysis, rector1991use, joyce2001introduction}.}
is already {well known} 
in seismic interferometry \citep{curtis2006seismic,vasconcelos2008interferometry1}.
To extract the response of a building,
\cite{snieder2006extracting} propose a deconvolution of the re\-corded 
waves at different locations
in the building 
rather than the cross-correlation.
%
%
Seismic interferometry by multi-dimen\-sional 
deconvolution \citep{wapenaar2008passive, wapenaar2011seismic,van2011controlled,broggini2014data}
uses an estimated interferometric point spread function as a deconvolution operator.
The results obtained from this approach depend on 
the accuracy of the estimated point spread function, which relies on 
a 
{uniform} distribution of 
multiple {noisy sources} in space.
In contrast to these  
seismic{-}interferometry{-}by{-}deconvolution approaches, 
FBD is designed to perform a \emph{blind} deconvolution in the presence of a single {noisy source} and 
doesn't assume an even distribution of the noisy sources.
In the presence of multiple {noisy sources}, 
as preprocessing to FBD,
one has to use
seismic blind source separation.
For example, \cite{makino2005blind} and \cite{bharadwaj2017deblending} used
independent component analysis for convolutive mixtures to 
decompose the multi-source recorded data into isolated records involving 
one source at a time.

The remainder of this paper is organized as follows.
We 
explain the indeterminacy of unregularized BD problem in section~\ref{sec:bd}.
In section~\ref{sec:fbd},
we introduce FBD and argue that it can
resolve
this indeterminacy.
This paper contains no theoretical guarantee,
but we regard the formulation of
such theorems 
as very interesting.
In section~\ref{sec:ex},
we demonstrate the benefits of FBD
using both idealized  
and practical synthetic  seismic experiments.
%
%
%


%
%

%
%
%
%

%
%

\section{Multichannel Blind Deconvolution}
\label{sec:bd}

The $z$-domain representations are denoted in this paper
using the corresponding capital letters. 
For example, the $i^\text{th}$ channel output
after a $z$-transform is denoted by 
\[D_i(z)=\sum_{t=0}^{T}d_i(t)z^{-t}.\] 
The traditional algorithmic
approach to solve BD 
is a least-squares fitting 
of the channel output vector $\vecc{d_i : \{0,\ldots,T\}\to \R}$ 
to jointly 
optimize two functions i.e., the impulse response vector
associated with different channels 
$\vecc{g_i : \{0,\ldots,\tau\}\to \R}$ and the source signature $s : \{0,\ldots,T\}\to \R$.
The joint optimization can be suitably carried out
using
alternating minimization
\citep{ayers1988iterative,sroubek2012robust}: 
in one cycle, we fix one function and optimize the other, and then fix the other and optimize the first. 
Several cycles are expected to be performed to reach convergence.
\medskip
\begin{defn}[{LSBD: Least-squares} Blind Deconvolution]\label{def:bd} 
	It is a
	basic formulation that minimizes the least-squares functional:
	\begin{flalign}
		\label{eqn:bd1}
		U(s,\vecc{g_i}) &=
		\sum_{k=1}^{\Nr} \sum_{t=0}^{T} \{d_{k}(t)-\{s\ast g_{k}\}(t)\}^2{;} \\
		\label{eqn:bd2}
		(\hat{s}, \vecc{\hat{g}_i}) &=  \mathop{\mathrm{arg\,min}}_{s,\vecc{g_i}} \qquad U \\
		  & \mathrm{subject \,\, to} \quad \sum_{t=0}^{T} s^2(t)=1{.} 
	\end{flalign}
	Here, $\hat{s}$ and $\vecc{\hat{g}_i}$ denote the predicted or estimated 
	functions 
	corresponding to the 
	unknowns 
	$s$ and
	$\vecc{g_i}$, respectively. 
	We have fixed the {energy (i.e., sum-of-squares)} norm of $s$ in order to 
	resolve the 
	scaling ambiguity.
	In order to effectively solve this
	problem,
	it is required 
	that the domain length $T+1$ of the first unknown function $s$
	be longer than the domain length $\tau+1$ of the second unknown function $\vecc{g_i}$ \citep{xu1995least}.
\end{defn} 

Ill-posedness is the major challenge of BD, 
irrespective of the number of channels.
For instance,
when the number of channels $\Nr=1$,
an undesirable minimizer for
eq.~\ref{eqn:bd1}
would be the temporal Kronecker $\delta(t)$ 
for the impulse response, 
making the source signature equal the channel output. 
Even with 
$\Nr\ge 1$, the LSBD problem 
can only be solved up to some indeterminacy.
To quantify the ambiguity, 
consider that a filter $\phi(t)\ne\delta(t)$
and its inverse 
$\phi^{-1}(t)$ ({where} $\phi \ast \phi^{-1} = \delta$)
can be applied to each element of 
$\vecc{g_i}$ and $s$ respectively, and leave their convolution unchanged:
\begin{eqnarray}
	d_{i}(t)=\{{s}\ast {g}_{i}\}(t)=\{\{{s}\ast\phi^{-1}\} \ast \{{g}_{i}\ast\phi\}\}(t){.}
	\label{eqn:decont2}
\end{eqnarray}
If furthermore ${s}\ast\phi^{-1}$ and $\vecc{{g}_{i}\ast\phi}$ obey the constraints otherwise 
placed on $s$ and $\vecc{g_i}$, namely in our case that 
$s$ and $\vecc{g_i}$ 
should have duration lengths $T+1$ and $\tau+1$ respectively, 
and the unity of the source energy,
then we are in 
presence of a true ambiguity not resolved by those constraints. 
We then speak of $\phi$ as belonging to a set $\mathbb{Q}$ of undetermined 
filters. 
This formalizes the lack of uniqueness \citep{xu1995least}: 
for any possibly desirable solution $(\hat{s}, \vecc{\hat{g}_i})$ 
and every $\phi\in\mathbb{Q}$,
$(\hat{s}\ast\phi^{-1}, \vecc{\hat{g}_{i}\ast\phi})$ 
is an additional {possibly} undesirable solution. 
Taking all $\phi\in\mathbb{Q}$ spawns all solutions in a set
$\mathbb{P}$ 
that equally minimize the 
least-squares functional in 
eq.~\ref{eqn:bd1}.
Accordingly, in the $z$-domain,
the elements in $\vecc{\hat{G}_i}$ of 
almost any solution
in $\mathbb{P}$
share some common root(s), which are
associated with its corresponding unknown filter $\Phi(z)$.
In other words, the channel polynomials in $\vecc{\hat{G}_i}$
of nearly all the solutions 
are \emph{not coprime}.
A particular element in $\mathbb{P}$ has 
its corresponding $\vecc{\hat{G}_i}$ with 
the fewest 
common roots ---
we call it the \emph{coprime} solution.
%
%

\section{Focused Blind Deconvolution}
\label{sec:fbd}
\medskip

\begin{figure*}
	\centering
	\IfFileExists{FIG/flowchart.png}
	{
	\includegraphics[width=0.8\textwidth]{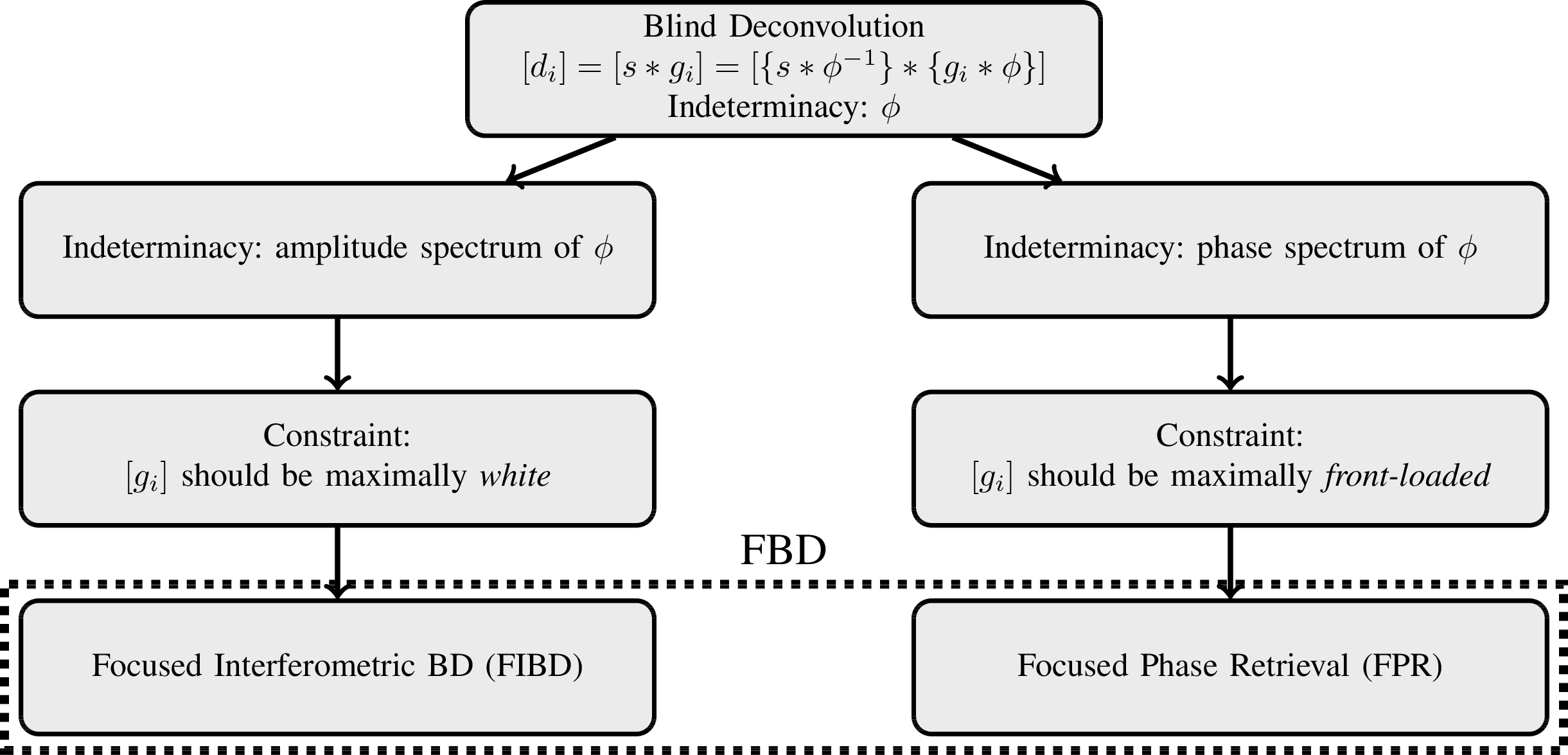}
		}
	{
	\input{./fbd_parts.tex}
	}
	\caption{Focused blind deconvolution uses two focusing constraints to resolve the 
	indeterminacies of the multichannel blind deconvolution. 
	Note that this is not an algorithmic flowchart, but
	explains the two components of the regularization in FBD.}
	\label{fig:fbd_parts}
\end{figure*}

The aim of focused blind deconvolution is to seek the coprime solution
of the
LSBD problem.
Otherwise, 
the channel polynomials $\vecc{\hat{G}_i}$ will typically be 
\emph{less sparse} and \emph{less front-loaded}
in the time domain owing to the common roots
that are 
associated with the
undetermined filter $\phi$ of eq.~\ref{eqn:decont2}.
For example,
including a common root $r$ to the 
polynomials in $\vecc{\hat{G}_i}$ implies an additional factor $(z-r)$ that 
corresponds to 
subtracting $\vecc{r\,g_i(t)}$ from $\vecc{g_i(t+1)}$ in the time domain,
so that the sparsity is likely to reduce.
Therefore, 
the intention and key innovation of FBD is to minimize the number of common roots 
in the channel polynomials $\vecc{\hat{G}_i}$ associated with $\Phi(z)$.
It is difficult to achieve the same result with standard 
ideas from sparse regularization.

%
%
Towards this end, 
focused blind deconvolution solves a series  
of two least-squares optimization problems
with
focusing constraints. 
These constraints, 
described in the following subsections, can guide FBD to converge to the desired coprime solution.
Note that
this prescription does not guarantee
that the recovered impulse responses 
should consistently 
match the true impulse responses;\footnote{In the seismic context, FBD does not guarantee 
that the recovered Green's function satisfies the wave equation with impulse source.}
nevertheless, we empirically 
encounter a 
satisfactory recovery in most practical situations 
of seismic inversion,
as discussed below.
The first problem considers
fitting 
the cross-correlated channel outputs
to jointly 
optimize two functions i.e., the impulse-response cross-correlations $\vecc{g_{ij}}$
between every possible channel pair
and the source-signature auto-correlation $\sa$.
The focusing constraint in this problem will resolve the indeterminacy due to the 
amplitude spectrum of the unknown filter  
$\phi$ in eq.~\ref{eqn:decont2} such that the impulse responses $\vecc{g_i}$ are \emph{maximally white}.
Then the second problem completes the focused blind deconvolution
by fitting 
the above-mentioned impulse-response cross-correlations, to estimate $\vecc{g_i}$ from $\vecc{g_{ij}}$.
The focusing constraint in this problem will resolve 
the indeterminacy due to the phase spectrum of 
the unknown filter
$\phi$ such that the impulse responses $\vecc{g_i}$ are \emph{maximally front-loaded}.
As shown in the Figure~\ref{fig:fbd_parts},
these two problems will altogether 
resolve the indeterminacies of BD discussed in
the previous section.

\subsection{Focused Interferometric Blind Deconvolution}

In order to isolate and resolve the indeterminacy due to the amplitude spectrum of $\phi(t)$, 
we consider a reformulated multichannel blind deconvolution problem.
This reformulation deals with
the 
cross-correlated or interferometric channel outputs, $d_{ij} : \{-T,\ldots,T\}\to \R$, 
as in eq.~\ref{eqn:intro1},  
between every possible
channel pair \citep[cf.,][]{demanet2017convex}, therefore 
ending the indeterminacy due to the phase spectrum of $\phi(t)$.
\begin{defn}[IBD: Interferometric Blind Deconvolution]\label{def:ibd} 
	We use this problem to lay the groundwork 
	for the next problem, and benchmarking.
	The optimization is carried {out} over the source-signature auto-correlation 
	$\sa : \{-T,\ldots,T\} \to \R$
	and the cross-correlated or interferometric impulse responses
	$g_{ij} : \{-\tau,\ldots,\tau\} \to \R${:}
	\begin{flalign}
		\label{eqn:ibd1}
		V(\sa,\vecc{g_{ij}}) &=
		\sum_{k=1}^{\Nr}\sum_{l=k}^{\Nr} \sum_{t=-T}^{T} \{d_{kl}(t)-\{\sa\ast g_{kl}\}(t)\}^2{;} \\
		\label{eqn:ibd2}
		(\sa[\hat], \vecc{\hat{g}_{ij}}) &=  \mathop{\mathrm{arg\,min}}_{\sa, \vecc{g_{ij}}} \qquad V \\ 
		  & \mathrm{subject \,\, to} \quad \sa(0)=1; \quad \sa(t)=\sa(-t){.} \nonumber
	\end{flalign}
	Here, we 
	denoted the 
	$(\Nr+1)\Nr/2$-vector
	of unique interferometric impulse responses 
	$\vecc{g_{11}, g_{12},\ldots,g_{22},g_{23},\ldots,g_{\Nr\Nr}}$ by simply 
	$\vecc{g_{ij}}$.
	We fit the 
	interferometric outputs $d_{ij}$ after max normalization. The motivation 
	of conveniently fixing  $\sa(0)$ is not only to resolve the scaling ambiguity
	but also 
	to converge to a solution, where the  
	necessary inequality condition $\sa(t) \le \sa (0)\,\forall\,t$ is satisfied.
	More generally,
	the function $\sa(t)$ is the autocorrelation of $s(t)$ 
	if and only if the Toeplitz matrix formed from its translates
	is positive semidefinite, i.e., $\mbox{Toeplitz}(\sa) \succeq 0$.
	This is a result known as Bochner's theorem.
	This semidefinite constraint 
	can 
	be realized by projecting $\mbox{Toeplitz}(\sa)$ 
	onto the cone of positive 
	semidefinite matrices at each iteration of the nonlinear least-squares iterative method \citep{vandenberghe1996semidefinite}.
	Nonetheless, in the numerical experiments, we observe
	convergence to acceptable solutions by just using the weaker constraints of IBD, when is data noise is sufficiently small.

	Similar to LSBD,
	IBD has unwanted minimizers 
	obtained by applying a filter $\psi^{-1}$ to $\sa$ and $\psi$ to 
	each element of $\vecc{g_{ij}}$, 
	but it is easily computed that $\psi$ 
	has to be \emph{real and nonnegative} in the frequency domain ($|z|=1$) 
	and related to the amplitude spectrum of $\phi(t)$.
	Therefore,
	its
	indeterminacy is lesser
	compared to that of the  
	LSBD approach.

\end{defn} 
\medskip

\begin{defn}[FIBD: Focused Interferometric Blind Deconvolution]\label{def:fibd} 
	{FIBD starts by seeking} a solution
	of the under\-determined IBD 
	problem {where} the 
	impulse responses are ``maximally white", 
	{as measured by the concentration of their autocorrelation {near} 
	zero lag}.
	Towards that end,
	we 
	use a {regularizing term} that penalizes the energy of 
	the impulse-response auto-correlations
	proportional to the
	non-zero lag time $t$,
	before returning to solving the regular IBD problem.
	\begin{flalign}
		\label{eqn:fibd1}
		W(\sa,\vecc{g_{ij}}) &=  
		V(\sa,\vecc{g_{ij}})
		+ \alpha 
		\sum_{k=1}^{\Nr} \sum_{t=-\tau}^{\tau}{t^2}g_{kk}^2(t){;} \\
		\label{eqn:fibd2}
		(\sa[\hat], \vecc{\hat{g}_{ij}}) &=  \mathop{\mathrm{arg\,min}}_{\sa, \vecc{g_{ij}}} \qquad W \\ 
		 &  \mathrm{subject\,to} \quad \sa(0)=1;\quad\sa(t)=\sa(-t){.} \nonumber
	\end{flalign}
	Here,
	$\alpha>0$ is a  regularization parameter.
	We consider a homotopy \citep{osborne2000new} approach to solve FIBD, 
	where {eq.}~\ref{eqn:fibd2} is solved {in succession for decreasing} values of $\alpha$, 
	{the result obtained for previous $\alpha$ {being} used as an initializer for the {cycle that uses} 
	the current $\alpha$.}
The focusing constraint 
resolves the indeterminac{y} of IBD.
Minimizing the energy of the  
impulse-response auto-correlations $\vecc{g_{ii}}$ 
proportional to the non-zero lag time 
will result in a solution where the impulse responses are 
{heuristically} as white as possible.
In other words, FIBD 
minimizes 
the number of common roots,  
associated with 
the IBD indeterminacy $\Psi(z)$,
in the estimated polynomials $\vecc{\hat{G}_{ij}}$, 
facilitating the goal of FBD to seek the coprime solution.
The entire 
workflow of FIBD is 
shown in the Algorithm~\ref{alg:fibd}.
	In most of the numerical examples, we simply choose
	$\alpha=\infty$ first, and then $\alpha=0$.
\end{defn}

%
%
%

\subsection{Focused Phase Retrieval}

FIBD resolves  
a component of 
the LSBD ambiguity and estimates the 
interferometric impulse responses.
This should be followed by
phase retrieval (PR) --- a least-squares fitting of the
interferometric impulse responses $\vecc{\hat{g}_{ij}}$ to optimize the impulse responses
$\vecc{g_i}$.
The estimation of $\vecc{g_i}$ in PR is hindered by the
unresolved LSBD ambiguity due to 
the phase spectrum of $\phi(t)$.
In order to
resolve the remaining ambiguity, 
we use a focusing constraint in PR.

\begin{defn}[LSPR: Least-squares Phase Retrieval]\label{def:lspr} 
	Given the interferometric impulse responses 
	$\vecc{g_{ij}}$, the aim of the 
	phase retrieval problem is to estimate unknown $\vecc{g_i}$.
	\begin{flalign}
	\label{eqn:lspr}
		X(\vecc{g_{i}}) &=
		\sum_{k=1}^{\Nr}\sum_{l=i}^{\Nr} \sum_{t=-\tau}^{\tau} \{\hat{g}_{kl}(t)-\{g_k \otimes g_l\}(t)\}^2{;} \\
		\vecc{\hat{g}_{i}}&=  \mathop{\mathrm{arg\,min}}_{\vecc{g_{i}}} \qquad X  
	\end{flalign}
%
LSPR 
	is ill-posed. Consider a \emph{white} filter $\chi(t)\ne\delta(t)$, where
$\chi \otimes \chi=\delta$, that can be
applied to each of the impulse responses, and leave their cross-correlations unchanged:
\begin{eqnarray}
	\label{eqn:lspr_unknown}
	g_{ij}(t) = \{g_i \otimes g_j\}(t) = \{\{g_i\ast\chi\}\otimes\{g_j\ast\chi\}\}(t).
\end{eqnarray}
%
If furthermore $g_i\ast\chi$ obeys the constraint otherwise placed, 
namely in our case that the impulse responses should have duration length 
$\tau$, then we are in the presence of a true ambiguity not resolved by this constraint.
It is obvious that 
the filter $\chi(t)$ is linked to the remaining unresolved component
of the LSBD indeterminacy, i.e., 
	the phase spectrum of $\phi(t)$. 
\end{defn}

\begin{defn}[FPR: Focused Phase Retrieval]\label{def:fpr} 
	FPR seeks a solution of the 
	underdetermined LSPR problem
	where the impulse responses $\vecc{g_i}$ are ``maximally front-loaded''.
	It starts with an optimization that 
	fits
	the interferometric impulse responses only linked with 
	the most front-loaded channel\footnote{In the seismic context, 
	the most front-loaded channel corresponds to the closest receiver
	$i=\inear$
	to the noisy source, assuming that the traveltime of the waves 
	propagating from the source to this receiver is the shortest.}
	$\inear$, 
	before returning to solving the regular LSPR problem.
	We use  
	a regularizing term that penalizes the energy of the most front-loaded 
	response $g_{\inear}$ proportional to the time $t\ne 0$:
	\begin{flalign}
	\label{eqn:fpr}
		Y(\vecc{{g}_i}) &= \sum_{k=1}^{\Nr}\sum_{t=-\tau}^{\tau} 
		\{\hat{g}_{k \inear}(t)-\{g_k \otimes g_\inear\}(t)\}^2  
		+ \beta\,\sum_{t=0}^{\tau} g_{\inear}^2(t)t^2; \\
		\vecc{\hat{g}_{i}} &=  \mathop{\mathrm{arg\,min}}_{g_{i}} \qquad Y.  
	\end{flalign}
	Here, $\beta\ge 0$ is a regularization parameter.
	Again, we consider a homotopy approach to solve this optimization problem,
	where the above equation is solved in succession for decreasing values of $\beta$.
FPR chooses the 
undetermined filter $\chi$ such that
	$g_i \ast \chi$ has the energy maximally concentrated or focused at the front (small $t$).
Minimizing  
the second moment of the squared impulse responses will result in a solution 
where the impulse responses are 
as front-loaded as possible.
The entire 
workflow of FPR is 
shown in the Algorithm~\ref{alg:fpr}.
In all the numerical examples, we simply choose
$\beta=\infty$ first, and then $\beta=0$.
Counting on the estimated impulse responses from FPR, we return to the LSBD formulation 
in order to finalize the BD problem.

\end{defn}

\subsection{Sufficiently Dissimilar Channel Configuration}
FBD seeks the coprime 
solution of the ill-posed LSBD problem.
Therefore, 
for the success of FBD, 
it is important 
that the true transfer functions
do not share any common zeros in the $z$-domain.
This requirement is satisfied when the 
channels are chosen
to be
\emph{sufficiently dissimilar}.
The channels
are said to be sufficiently dissimilar 
unless there exists a spurious $\gamma$ 
and $\vecc{g_i}$ such that 
the true impulse-response vector $\vecc{g^0_{i}} = \vecc{\gamma\ast g_{i}}$.
Here, $\gamma$ is a 
filter that 
\begin{inparaenum}
\item is independent of the channel index $i$;
\item belongs to the set $\mathbb{Q}$ of filters that cause indeterminacy of the LSBD problem;
\item doesn't simply shift $g_i$ in time.
\end{inparaenum}
{In our experiments}, FBD reconstructs {a good approximation of} the true 
impulse responses
if the channels are sufficiently dissimilar.
Otherwise, FBD outputs
an 
undesirable solution
$(s^0\ast\gamma^{-1},\vecc{g_{i}})$,
as opposed to the desired $(s^0,\vecc{\gamma\ast g_{i}})$, where $s^0$ is the true 
source signature.
In the next section,
we will show numerical examples 
with both similar and dissimilar channels.

\LinesNumbered
\begin{algorithm2e} 
	\caption{Focused Interferometric Blind Deconvolution. Alternating minimization of $W$, as in
	{eq.}~\ref{eqn:fibd2}, is carried out {in succession for decreasing} values of $\alpha$.}
	\label{alg:fibd}
	\begin{itemize}
		\item[{\bf Preparation}]
		{\color{white} gg}\\
			generate the cross-correlated or interferometric channel outputs
			$\vecc{d_{ij}}$ and normalize with $d_{11}(0)$
		\item[{\bf Parameters (with example)}]
		{\color{white} gg}\\
		tolerance for convergence $\epsilon=10^{-8}$\\
			$\vec{\alpha}=\set{\infty,0}$
	\item[{\bf Initialize}] 
		{\color{white} gg}\\
			$\sa(t)\leftarrow\begin{cases}
				0, &\text{if } t\ne0\\
				1, &\text{otherwise}
			\end{cases}
				$\\
			$g_{ij}(t)\leftarrow
			\begin{cases}
				0, &\text{if } i=j\,\text{and }t\ne0 \\
				\operatorname{rand}(), & \text{otherwise}
			\end{cases}
			$\\
		\item[{\bf Results}]
		{\color{white} gg}\\
			interferometric transfer functions $\vecc{\hat{g}_{ij}}$ \\
			autocorrelation for the source signature $\hat{s}_a$
	\end{itemize}
	{\bf Kernel}\\
	\ForEach(\tcc*[f]{loop over decreasing $\alpha$}){$\alpha\in{\vec{\alpha}}$}{
		$W_1=\infty$; $W_2=\infty$; $W_{1p}=W_1$; $W_{2p}=W_2$; $\Delta W=\infty$\\
		\While{$\Delta W>\epsilon$} 
	{

		{${s}_a \leftarrow  \operatorname*{arg\,min}\limits_{s_{a}} W(\sa, \vecc{{g}_{ij}})$ s.t. $\sa(0)=1$ \& $\sa(t)=\sa(-t)$ }
		\tcc*[f]{updating source}\\
	$W_{1p} \leftarrow W_1$;
	     {$W_1\leftarrow W({s}_a, \vecc{{g}_{ij}})$} \\
	{$\vecc{{g}_{ij}} \leftarrow  \operatorname*{arg\,min}\limits_{\vecc{g_{ij}}} W({s}_a, \vecc{g_{ij}})$}
	
	\tcc*[f]{updating interferometric transfer functions}\\
	$W_{2p} \leftarrow W_2$;
	     {$W_2\leftarrow W({s}_a, \vecc{{g}_{ij}})$} \\

	     {$\Delta W=\operatorname{max}(\set{W_{1p}-W_1, W_{2p}-W_2}$}) \tcc*[f]{measure convergance}
	     }
	     }
	{$\vecc{\hat{g}_{ij}}\leftarrow \vecc{g_{ij}$}};$\,$
	{$\hat{s}_{a}\leftarrow s_{a}$}
\end{algorithm2e}

\LinesNumbered
\begin{algorithm2e} 
	\caption{Focused Phase Retrieval. Solving $Y$, as in {eq.}~\ref{eqn:fpr}, in succession for decreasing values of $\beta$.
	Then solving $X$ in {eq.}~\ref{eqn:lspr}.}
	\label{alg:fpr}
	\begin{itemize}
		\item[{\bf Preparation}]
		{\color{white} gg}\\
			get the interferometric filters $\vecc{\hat{g}_{ij}}$ using FIBD
		\item[{\bf Parameters (with example)}]
		{\color{white} gg}\\
			$\vec{\beta}=\set{\infty,0}$\\
			index of the most front-loaded channel $\inear$
	\item[{\bf Initialize}] 
		{\color{white} gg}\\
			$g_{i}(t)\leftarrow
			\begin{cases}
				0, &\text{if } i=\inear\,\text{and }t\ne0 \\
				\operatorname{rand}(), & \text{otherwise}
			\end{cases}
				$\\
		\item[{\bf Results}]
		{\color{white} gg}\\
			filters $\vecc{\hat{g}_{i}}$ \\
	\end{itemize}
	{\bf Kernel}\\
	\ForEach(\tcc*[f]{loop over decreasing $\beta$}){$\beta\in{\vec{\beta}}$}{

		$\vecc{{g}_i} \leftarrow  \operatorname*{arg\,min}\limits_{\vecc{g_{i}}} Y(\vecc{{g}_{i}})$\\
	     }
{ 
	$\vecc{{g}_{i}} \leftarrow  \operatorname*{arg\,min}\limits_{\vecc{g_{i}}} X(\vecc{g_{i}})$
	     }\tcc*[f]{return to LSPR}

	{$\vecc{\hat{g}_{i}}\leftarrow \vecc{g_{i}}$}
\end{algorithm2e}

\section{Numerical Simulations}
\label{sec:ex}

\subsection{Idealized Experiment I}
We consider an
experiment with $\Nr=20$, $\tau=30$ and $T=400$.
The aim is to
reconstruct the true impulse responses $\vecc{g^{0}_{i}}$,
plotted in Figure~\ref{fig:simple_example_bd0}a, 
from the channel outputs generated using a 
Gaussian random source signature $s^{0}$.
The impulse responses of 
similar kind  
are of 
particular interest in seismic inversion and room acoustics 
as they reveal the arrival
of 
energy,
propagated from an impulsive source,
at the receivers in the medium.
In this case, the arrivals have onsets of 6$\,$s and 10$\,$s at the first channel and they 
curve 
linearly and hyperbolically, respectively.
The linear arrival is 
the earliest arrival that doesn't undergo scattering.
The hyperbolic arrival is likely to represent a wave that is 
reflected or scattered from an interface between two materials with different acoustic impedances.
\subsubsection*{LSBD}

\begin{figure*}
	\begin{center}
		\IfFileExists{./FIG/simple_example_bd0.png}
		{
		\includegraphics[width=0.95\textwidth]{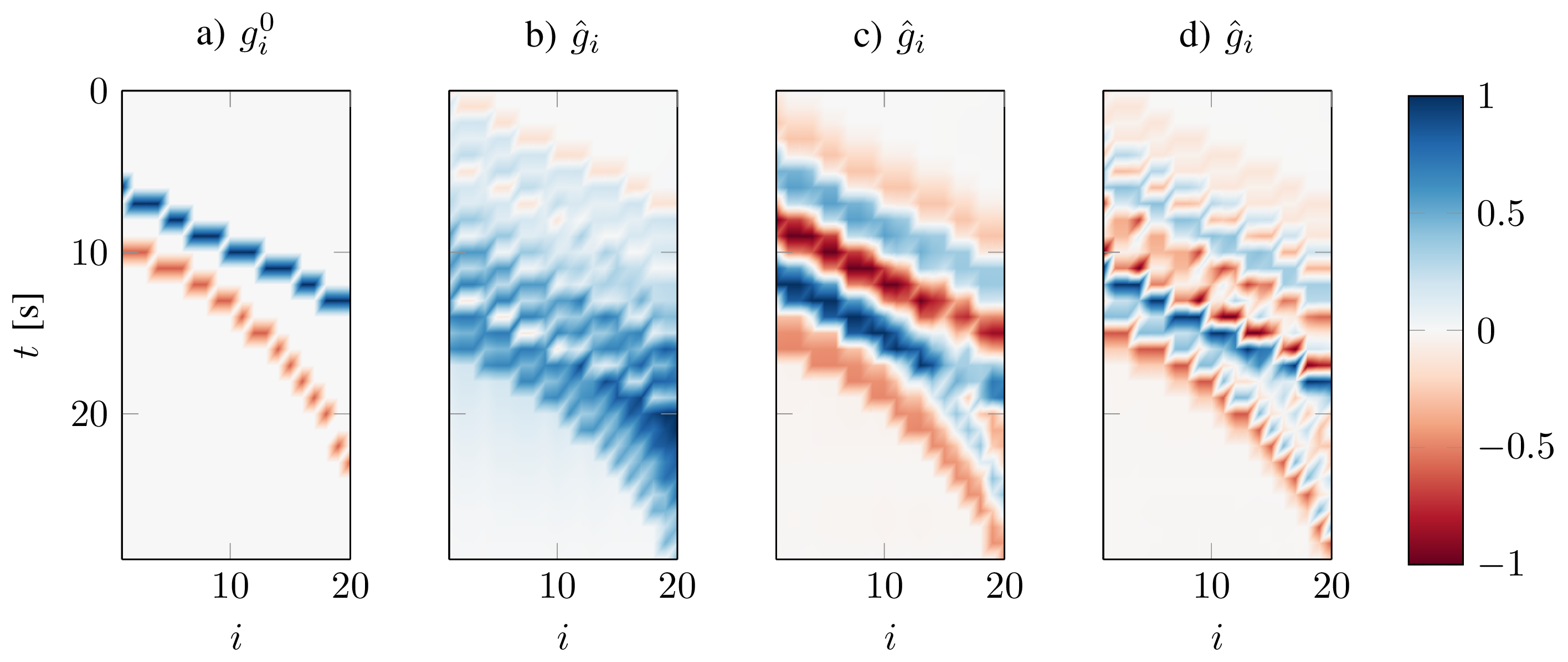}
			}
		{
	\begin{tikzpicture}
		\input{./codes/simple_example_bd0.tex}
	\end{tikzpicture}
		}
	\end{center}
	\caption{
		Idealized Experiment I.
	 The results are displayed as images that use 
	 the full range of colors in a colormap. 
	 Each pixel of these images corresponds to a time $t$ and a channel index $i$.
	 Impulse responses: a) true; b)---d) undesired.
	 }
	\label{fig:simple_example_bd0}
\end{figure*}

\begin{figure*}
	\begin{center}
		\IfFileExists{./FIG/simple_example_bd1.png}
		{
		\includegraphics[width=0.95\textwidth]{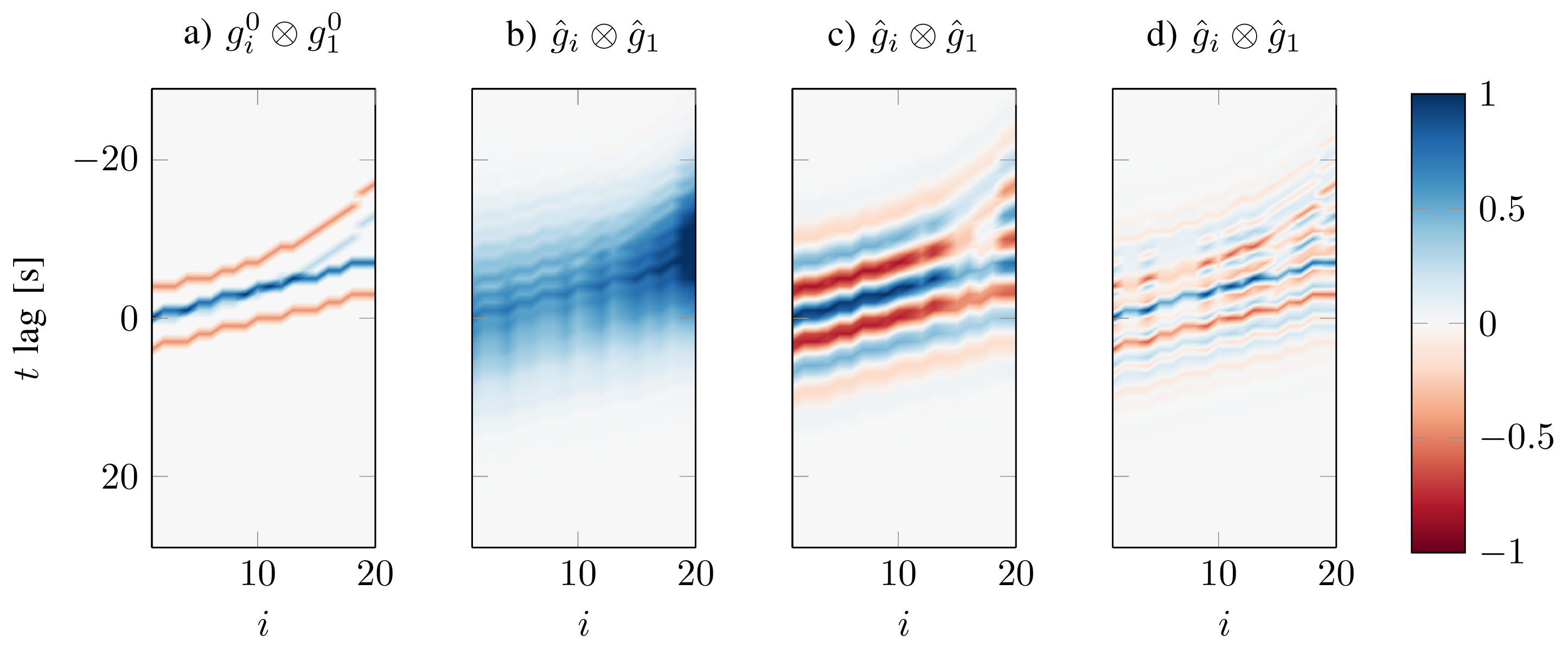}
			}
		{
	\begin{tikzpicture}
		\input{./codes/simple_example_bd1.tex}
	\end{tikzpicture}
		}
	\end{center}
	\caption{
		Idealized Experiment I.
		Cross-correlations of impulse responses corresponding to the Figure~\ref{fig:simple_example_bd0}:
	 a) true; b)---d) undesired.
	 }
	\label{fig:simple_example_bd1}
\end{figure*}

\begin{figure*}
	\centering
	\IfFileExists{./FIG/simple_example_fbd.png}
	{
	\includegraphics[width=0.95\textwidth]{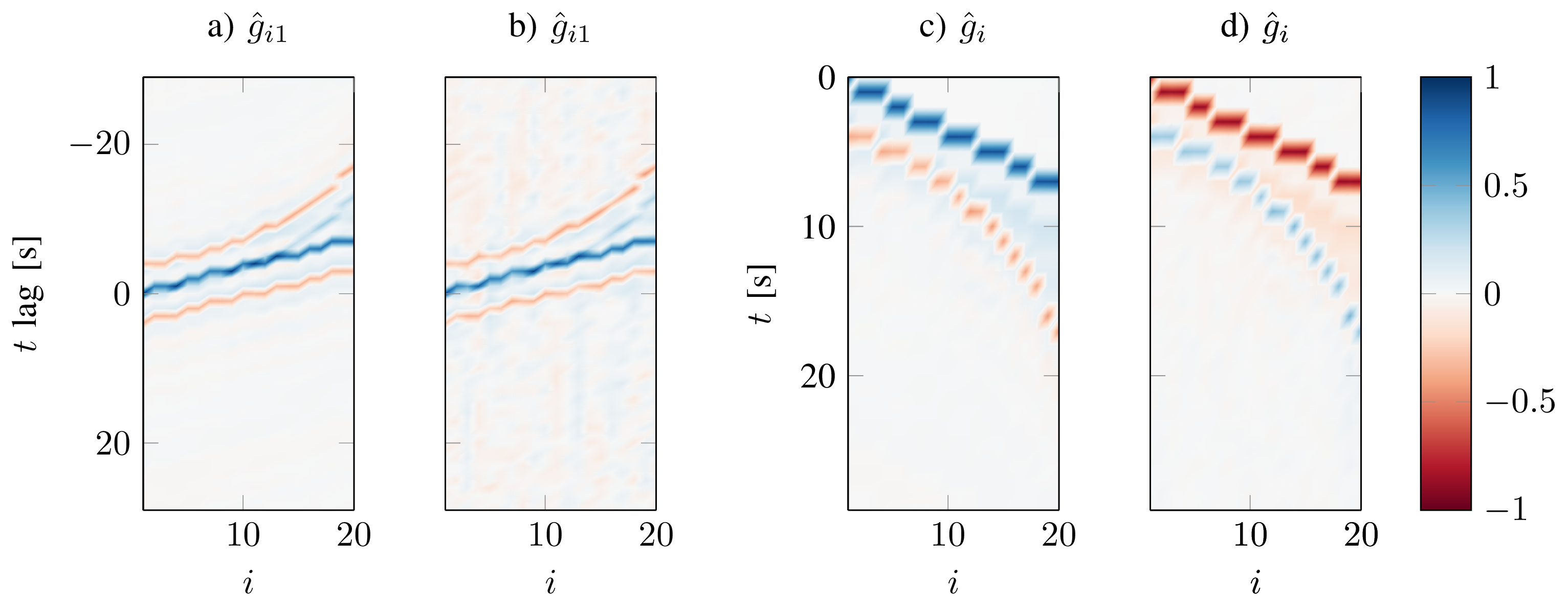}
		}
	{
	\begin{tikzpicture}
		\input{./codes/simple_example_fbd.tex}
	\end{tikzpicture}
	}
	\caption{ Idealized Experiment I. 
		a) FIBD estimated interferometric 
		impulse responses corresponding to the Figure~\ref{fig:simple_example_bd1}a, 
		after fitting the
		interferometric channel outputs.
		b) Same as (a), except after white noise is added to the channel outputs.
		c) Estimated impulse responses from FPR by fitting the FIBD-outcome interferometric impulse responses in (a).
		d) Same as (c), except fitting the FIBD outcome in (b).
	}
	\label{fig:simple_example_fbd}
\end{figure*}

\begin{figure*}
	\begin{center}
		\IfFileExists{./FIG/simple_example_bd2.png}
		{
		\includegraphics[width=0.95\textwidth]{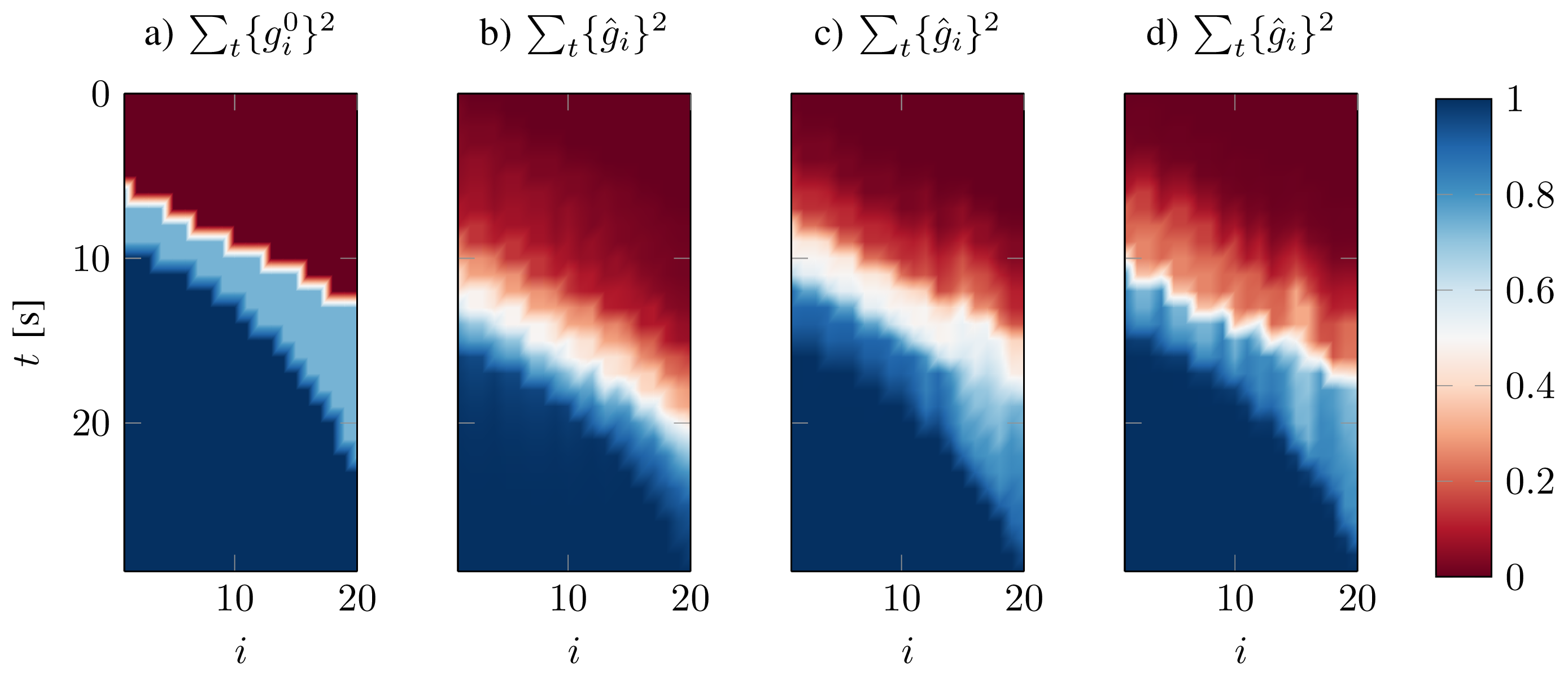}
			}
		{
	\begin{tikzpicture}
		\input{./codes/simple_example_bd2.tex}
	\end{tikzpicture}
		}
	\end{center}
	\caption{
		Idealized Experiment I.
	 Normalized cumulative energy of: a) true; b)---d) undesired impulse responses corresponding to the Figure~\ref{fig:simple_example_bd0}.
	 }
	\label{fig:simple_example_bd2}
\end{figure*}

To illustrate its non-uniqueness, we use 
three different initial estimates of 
$s$ and $\vecc{g_i}$ to observe the convergence to 
three different solutions that belong to $\mathbb{P}$.
The channel responses corresponding to these solutions are plotted in 
Figures~\ref{fig:simple_example_bd0}b--d.
At the convergence, the misfit (given in eq.~\ref{eqn:bd1}) 
in all these three cases $U(\hat{s},\vecc{\hat{g}_i})\lessapprox 10^{-6}$, justifying non-uniqueness.
Moreover, we notice that none of the 
solutions 
is desirable due to 
insufficient resolution.

\subsubsection*{FIBD}

In order to isolate the indeterminacy due to the 
amplitude spectrum of the unknown filter $\phi(t)$ in eq.~\ref{eqn:decont2}
and 
justify the use of the 
focusing constraint in eq.~\ref{eqn:fibd1},
we plot 
the true and undesirable impulse responses after cross-correlation
in the Figure~\ref{fig:simple_example_bd1}.
It can be easily noticed
that the true impulse-response cross-correlations corresponding to the first channel
are more focused at $t=0$ 
than the 
undesirable impulse-response cross-correlations.
The defocusing is caused by the ambiguity 
related to the amplitude spectrum of $\phi(t)$.
FIBD in Algorithm~\ref{alg:fibd} with $\vec{\alpha}=[\infty,0.0]$ 
resolves this ambiguity 
and
satisfactorily recovers the true
interferometric impulse responses $\vecc{g^0_{ij}}$, as plotted in Figure~\ref{fig:simple_example_fbd}a.
We regard the FIBD recovery be satisfactory in Figure~\ref{fig:simple_example_fbd}b
when the 
Gaussian white noise is added to the channel outputs so that the signal-to-noise (SNR) is $1\,$dB.

\subsubsection*{FPR}

In order to motivate the use of the second focusing constraint,
we plotted the 
normalized 
cumulative energy 
of the 
true and undesired impulse responses
in the Figures~\ref{fig:simple_example_bd2}.
It can be easily noticed that the fastest 
rate of energy buildup in time occurs in the 
case of the true impulse responses.
In other words, the energy of the 
true impulse responses is more front-loaded compared to 
undesired impulse responses, after neglecting an overall translation in time.
The FPR in Algorithm~\ref{alg:fpr} with $\vec{\beta}=[\infty, 0]$ 
satisfactorily recovers $\vecc{g^0_{i}}$ that are plotted in: the 
Figure~\ref{fig:simple_example_fbd}c --- utilizing $\vecc{g_{ij}}$ recovered from the noiseless channel outputs (Figure~\ref{fig:simple_example_fbd}a);
the Figure~\ref{fig:simple_example_fbd}d --- utilizing $\vecc{g_{ij}}$ 
recovered from the channel outputs (Figure~\ref{fig:simple_example_fbd}b) with Gaussian white noise.
Note that the overall time translation and scaling cannot be fundamentally determined.

\subsection{Idealized Experiment II}
This IBD-benchmark 
experiment with $\Nr=20$ $\tau=30$ and $T=400$
aims to reconstruct simpler  
interferometric impulse responses, plotted in 
Figure~\ref{fig:simple_example_ibd_fail}b, corresponding to the 
true impulse responses in Figure~\ref{fig:simple_example_ibd_fail}a.
A satisfactory recovery of $\vecc{g^0_{ij}}$ is not achievable without the focusing constraint ---
the IBD outcome in the Figure~\ref{fig:simple_example_ibd_fail}c
doesn't match the true 
interferometric impulse responses in the Figure~\ref{fig:simple_example_ibd_fail}b, unlike FIBD in the Figure~\ref{fig:simple_example_ibd_fail}d.

\begin{figure*}
	\centering
	\IfFileExists{./FIG/simple_example_ibd_fail.png}
	{
	\includegraphics[width=0.95\textwidth]{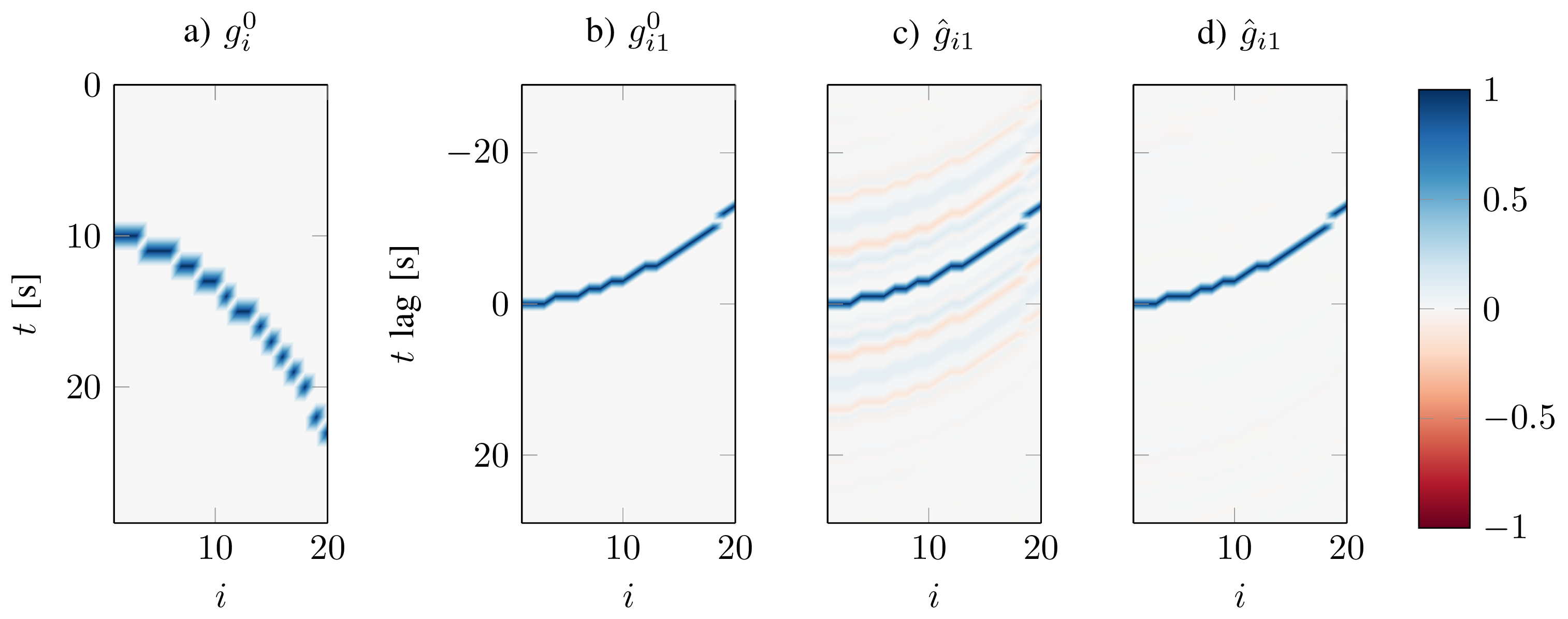}
		}
	{
	\begin{tikzpicture}
		\input{./codes/simple_example_ibd_fail.tex}
	\end{tikzpicture}
	}
	\caption{Idealized Experiment II.
	Interferometric impulse responses:
	a) true; 
	b) estimated using IBD; c) estimated using FIBD.}
	\label{fig:simple_example_ibd_fail}
\end{figure*}

\subsection{Idealized Experiment III}
We consider another 
experiment with $\Nr=20$ and $\tau=30$
to
reconstruct the true impulse responses $\vecc{g^{0}_{i}}$ 
(plotted in Figure~\ref{fig:simple_example_pr_fail}a) 
by fitting their
cross-correlations
$\vecc{g^0_{ij}}$.
A satisfactory recovery of $\vecc{g^0_{i}}$ from $\vecc{g^0_{ij}}$ is not achievable without the focusing constraint ---
the outcome of LSPR, in Figure~\ref{fig:simple_example_pr_fail}b,
doesn't match the true impulse responses, in Figure~\ref{fig:simple_example_pr_fail}a, but is 
contaminated by the filter $\chi(t)$ in eq.~\ref{eqn:lspr_unknown}.
On the other hand, FPR results in the outcome (Figure~\ref{fig:simple_example_pr_fail}c) 
that is not contaminated by $\chi(t)$.
\begin{figure*}
	\centering
	\IfFileExists{./FIG/simple_example_pr_fail.png}
	{
	\includegraphics[width=0.95\textwidth]{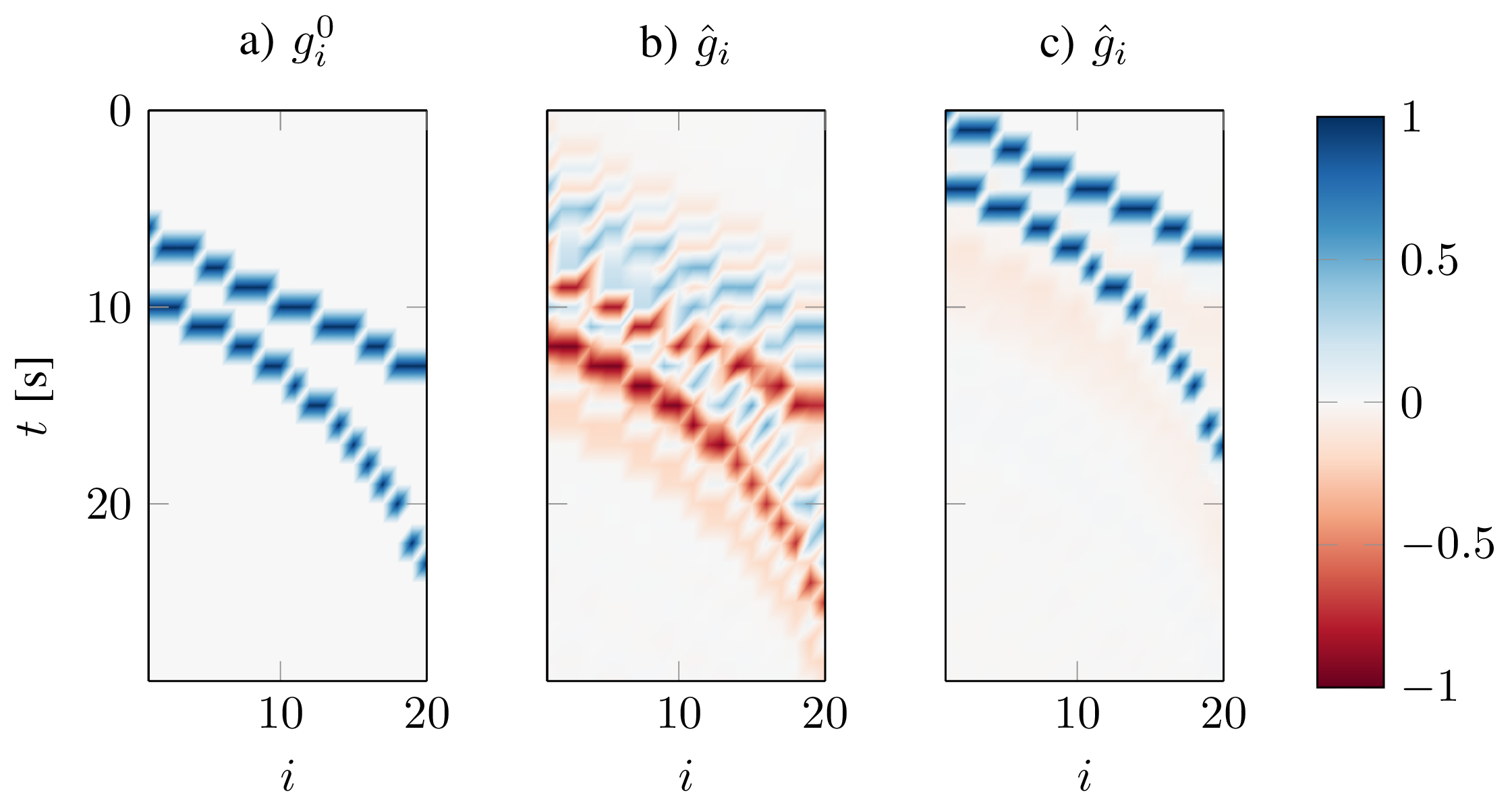}
		}
	{
	\begin{tikzpicture}
		\input{./codes/simple_example_pr_fail.tex}
	\end{tikzpicture}
	}
	\caption{Idealized Experiment III.
	a) True impulse responses. 
	b) Estimated impulse responses using LSPR.
	c) Estimated impulse responses using FPR.
	}
	\label{fig:simple_example_pr_fail}
\end{figure*}

\subsection{Idealized Experiment IV}
This 
experiment with $\Nr=20$,
$\tau=30$ and $T=400$  
aims to reconstruct the true 
interferometric impulse responses, plotted in 
Figure~\ref{fig:simple_example_fibd_fail}b, corresponding to the 
true impulse responses in Figure~\ref{fig:simple_example_fibd_fail}a.
The outcome of 
FIBD with $\vec{\alpha}=[\infty,0]$, plotted in Figure~\ref{fig:simple_example_fibd_fail}c,
doesn't clearly match the true 
interferometric impulse responses
because the channels are \emph{not} sufficiently dissimilar.
In this regard,
observe that the Figure-\ref{fig:simple_example_fibd_fail}a
true impulse responses 
at various channels 
$i$ differ only by a fixed time-translation instead 
of curving as in Figure~\ref{fig:simple_example_bd0}a.
\begin{figure*}
	\centering
	\IfFileExists{./FIG/simple_example_fibd_fail.png}
	{
	\includegraphics[width=0.95\textwidth]{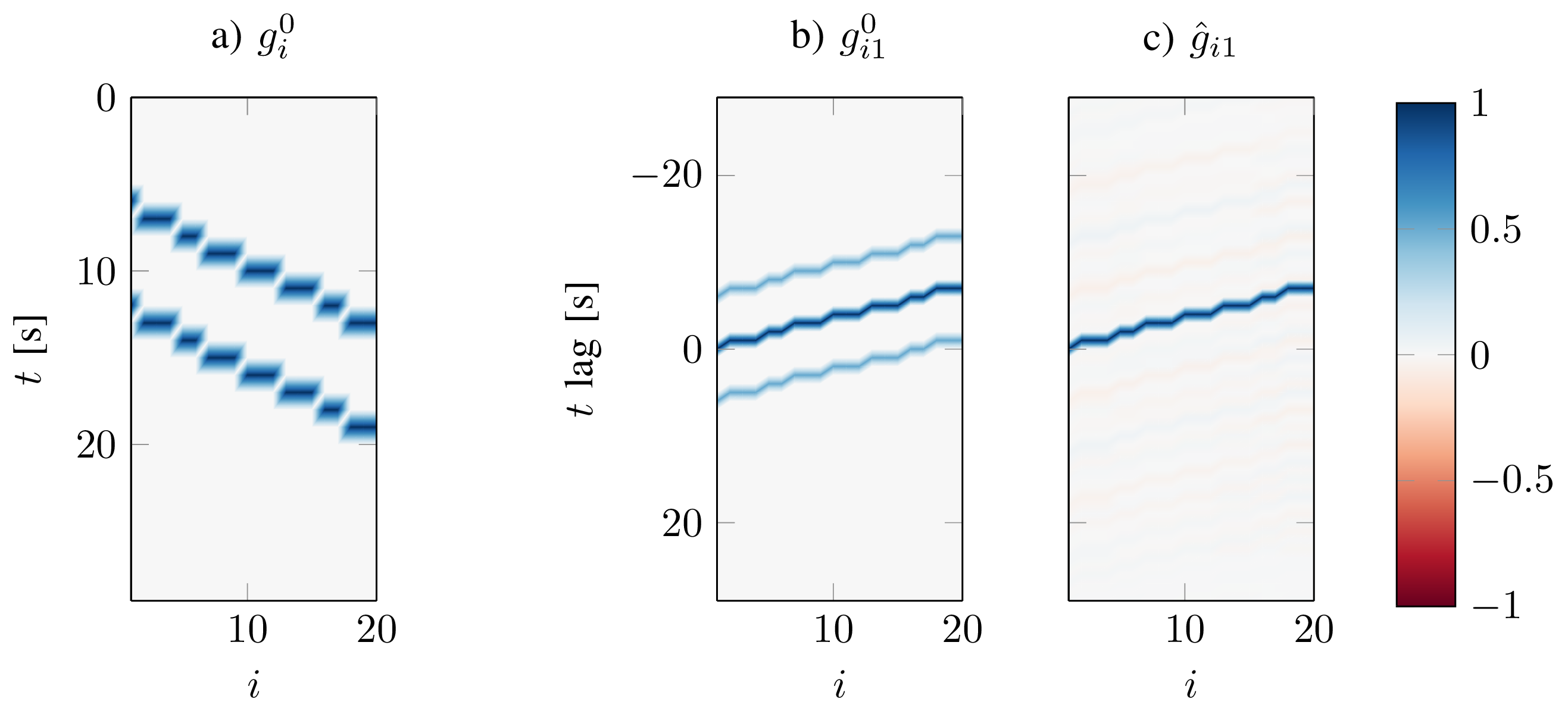}
		}
	{
	\begin{tikzpicture}
		\input{./codes/simple_example_fibd_fail.tex}
	\end{tikzpicture}
	}
	\caption{ Idealized Experiment IV. 
		a) True impulse responses of channels that are not sufficiently dissimilar.
		b) True interferometric impulse responses corresponding to (a).
		c) FIBD estimated interferometric 
		impulse responses corresponding to (b), 
		after fitting the
		interferometric channel outputs.
	}
	\label{fig:simple_example_fibd_fail}
\end{figure*}

\subsection{Idealized Experiment V}
We consider another 
experiment with $\Nr=20$ 
and $\tau=30$
to reconstruct the true impulse responses $\vecc{g^{0}_{i}}$ 
(plotted in Figure~\ref{fig:simple_example_fpr_fail}a) that are 
\emph{not} front-loaded, by fitting their
cross-correlations
$\vecc{g^0_{ij}}$.
The FPR estimated impulse responses $\vecc{\hat{g}_{i}}$, plotted in 
Figure~\ref{fig:simple_example_fpr_fail}b, 
do not 
clearly depict the arrivals 
because there exists a spurious
$\chi\ne\delta$ obeying 
eq.~\ref{eqn:lspr_unknown}, such that $\vecc{g^0_{i}\ast \chi}$ are
more front-loaded than $\vecc{g^0_{i}}$.
We observe that
FPR typically doesn't result in a favorable outcome 
if the impulse responses are not front-loaded.
Otherwise, 
the 
front-loaded 
$\vecc{g^{0}_{i}}$, plotted in Figure~\ref{fig:simple_example_fpr_fail}c, 
are successfully reconstructed in Figure~\ref{fig:simple_example_fpr_fail}d,
except 
for an overall translation in time.

\begin{figure*}
	\centering
	\IfFileExists{./FIG/simple_example_fpr_fail.png}
	{
	\includegraphics[width=0.95\textwidth]{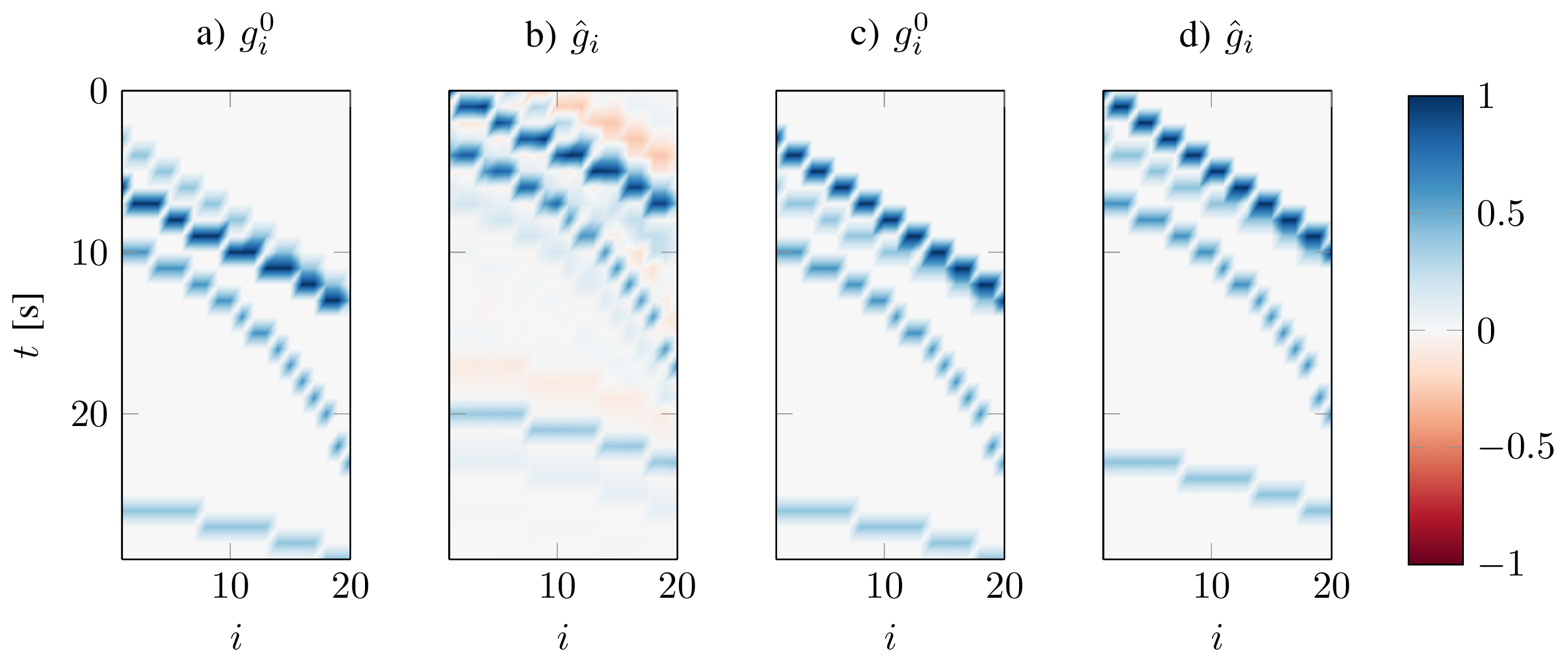}
		}
	{
	\begin{tikzpicture}
		\input{./codes/simple_example_fpr_fail.tex}
	\end{tikzpicture}
	}
	\caption{
		Idealized Experiment V.
		a) True impulse responses that are not front-loaded.
		b) FPR estimated impulse responses corresponding to (a), 
		after fitting the
		true interferometric impulse responses.
		c) Same as (a), but front-loaded.
		d) Same as (b), but corresponding to (c).
	}
	\label{fig:simple_example_fpr_fail}
\end{figure*}

\section{Green's function Retrieval}

\begin{figure*}
	\centering
	\centering
	\IfFileExists{./FIG/wav.png}
	{
	\includegraphics[width=0.5\textwidth]{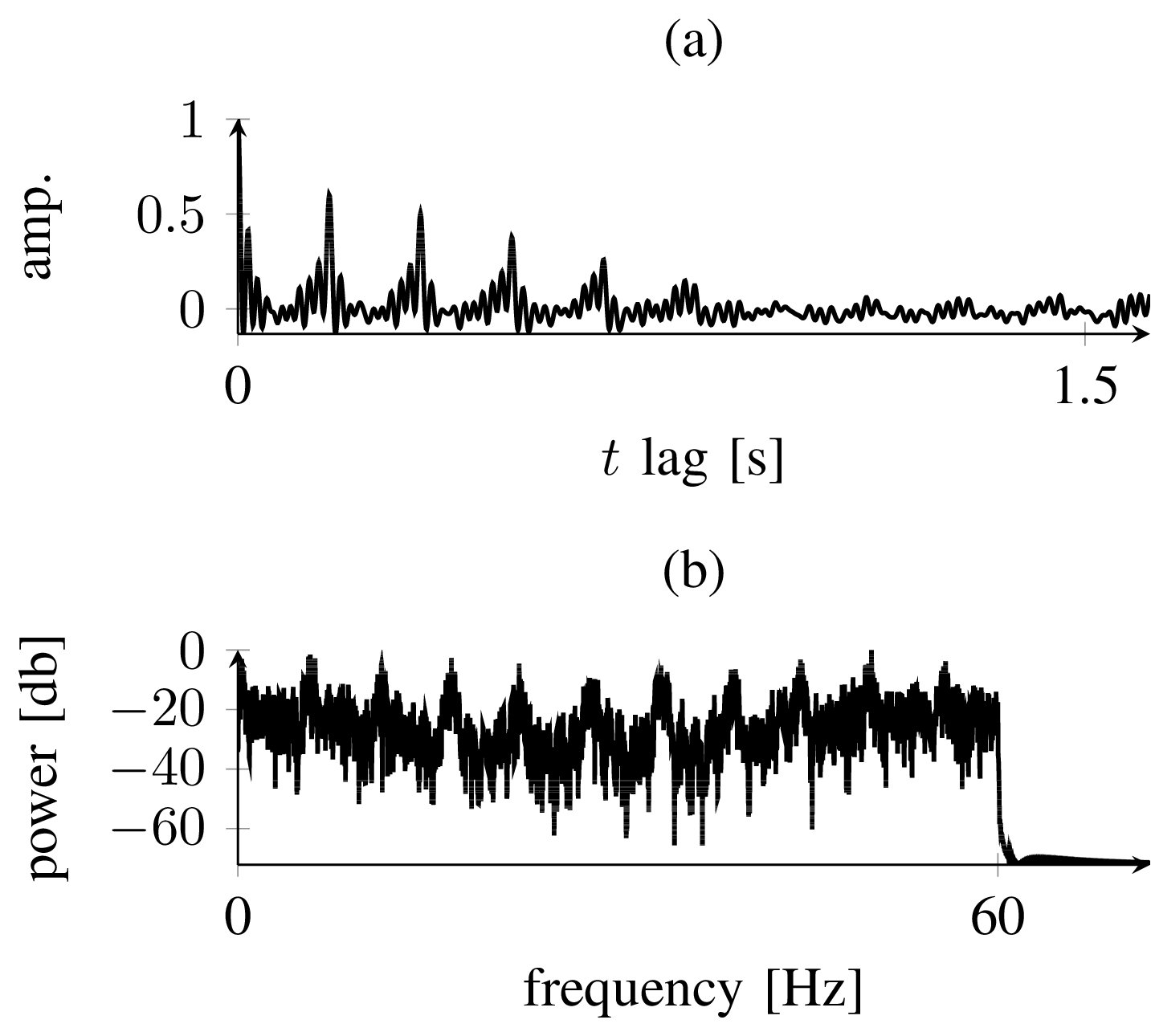}
		}
		{
	\input{codes/main_wavobs.tex}
	}
 \caption{
	 Source signature for the seismic experiment.
	 (a) auto-correlation that contaminates the interferometric Green's functions in the time domain ---
	 only 5\% of $T$ is plotted;
	 (b) power spectrum, where the Nyquist frequency is $60\,$Hz.
 }
 \label{fig:wavobs}
\end{figure*}
%
%

\begin{figure*}
	\centering
	\IfFileExists{./FIG/marmousi.png}
	{
	\includegraphics[width=0.95\textwidth]{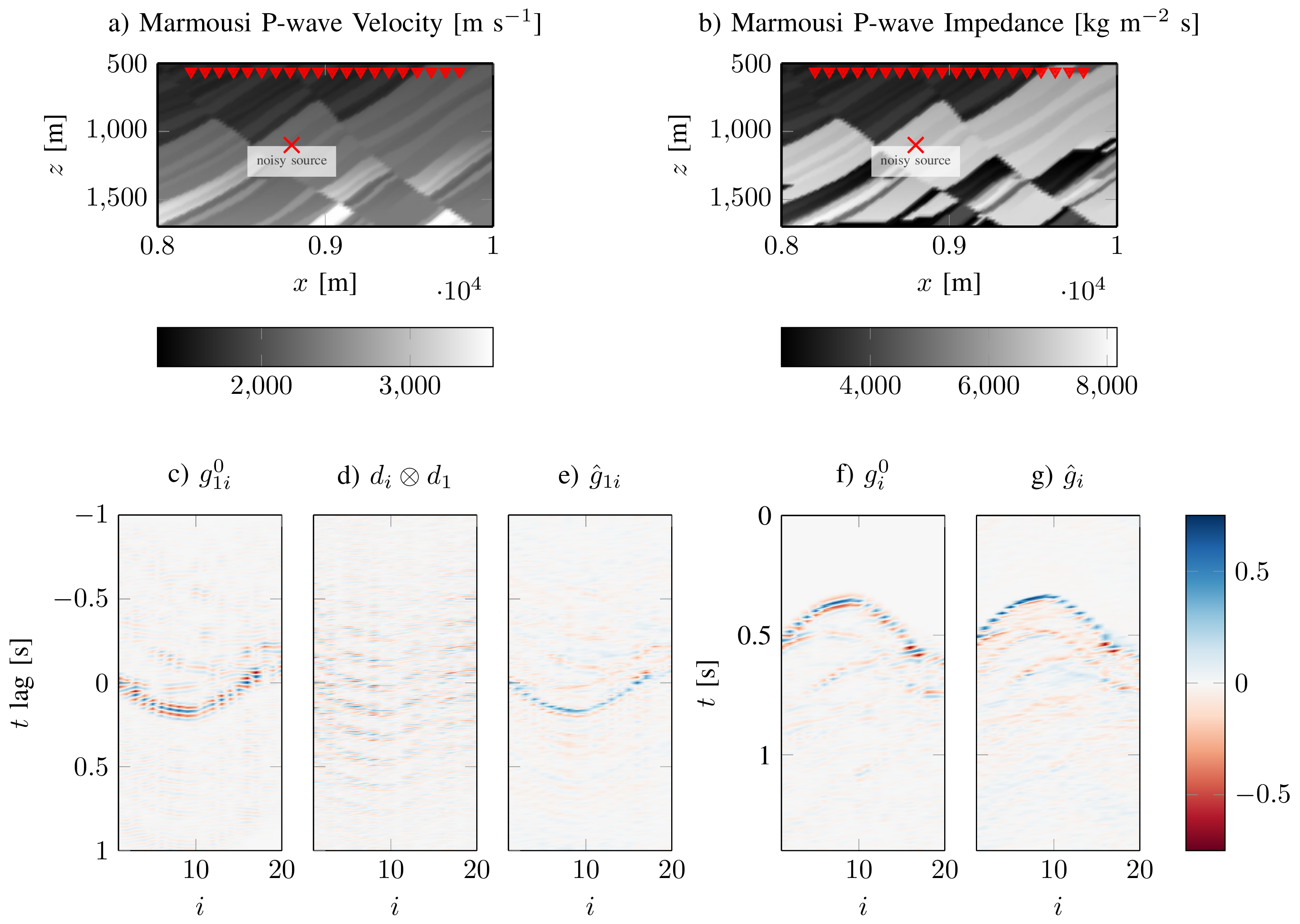}
		}
	{
	\begin{tikzpicture}
		\input{./codes/marmousi.tex}
	\end{tikzpicture}
	}
 \caption{
	 Seismic Experiment.
	 a) Acoustic velocity model for wave propagation.
	 b) Acoustic impedance model depicting interfaces that reflect waves.
	 c) True interferometric Green's functions.
	 d) Seismic interferometry by cross-correlation.
	 e) FIBD estimated interferometric Green's functions.
	 f) True Green's functions.
	 g) FBD estimated Green's functions.
 }
 \label{fig:reflec_main_marm}
\end{figure*}

Finally, we consider a more realistic scenario 
involving seismic-wave propagation
in a 
complex 2-D structural model, which is 
commonly known as the Marmousi model \citep{brougois1990marmousi} 
in exploration seismology. 
The Marmousi 
P-wave velocity and impedance plots are  
in the Figures~\ref{fig:reflec_main_marm}a and \ref{fig:reflec_main_marm}b, respectively.
We inject an unknown band-limited source signal, e.g., due to a drill bit,
into this model for $30\,$s, such that $T=3600$. 
The signal's auto-correlation 
and power spectrum are plotted in Figures~\ref{fig:wavobs}a and \ref{fig:wavobs}b, respectively.
We used an acoustic time-domain 
staggered-grid finite-difference solver for wave-equation modeling.
The recorded seismic data 
at twenty 
receivers spaced roughly $100\,$m apart, placed at a depth of roughly $500\,$m,
can be modeled as
the output of a linear system that convolves the source signature
with the Earth's
impulse response, i.e., its Green's function.
We recall that
in the seismic context:
\begin{itemize}

	\item the impulse responses $\vecc{g_i}$ correspond 
	to the 
	unique subsurface Green's function {$g(\vec{x},t)$} 
	evaluated at the receiver locations {$\vecc{\vec{x}_i}$},
	where the seismic-source signals are recorded;

	\item the channel outputs $\vecc{d_i}$ correspond to 
		the noisy subsurface wavefield $d(\vec{x},t)$ 
	recorded at the receivers 
		only for $\{0,\ldots,T\}$ ---  
		we are assuming that the source may be arbitrarily on or off 
	throughout this time interval, just as in usual drilling operations;

	\item $\tau$ denotes the propagation time
	necessary for the 
	seismic energy,
	including multiple scattering, 
	traveling 
	from the source to a total of $\Nr$ receivers,
	to decrease below an ad-hoc threshold.
\end{itemize}
The goal of 
this experiment
is to reconstruct the subsurface Green's function vector $\vecc{g_{i}}$ 
that contains:
\begin{inparaenum}
	\item the direct arrival from the source to the receivers and
\item   the scattered waves from
	various interfaces in the model.
\end{inparaenum}
The `true' Green's functions $g^0_{i}$ and the  
interferometric Green's functions $g^0_{ij}$,
in Figures~\ref{fig:reflec_main_marm}f and \ref{fig:reflec_main_marm}c,
are generated following these steps: 
\begin{inparaenum}
\item get data for $1.5\,$s ($\tau=180$) using a Ricker source wavelet (basically a degree-2 Hermite function modulated to a 
	peak frequency of 20$\,$Hz);
\item create cross-correlated data necessary for $\vecc{g^0_{ij}}$; {and}
\item perform a deterministic deconvolution on the data 
	using the Ricker wavelet.
\end{inparaenum}
Observe that we have chosen
the 
propagation time to be 
$1.5\,$s, such that $T/\tau=20$.

Seismic interferometry by cross-correlation (see eq.~\ref{eqn:intro1})
fails to retrieve 
direct and the scattered arrivals in the 
true interferometric
Green's functions,
as the  
cross-correlated data $\vecc{d_{ij}}$, plotted in Figure~\ref{fig:reflec_main_marm}d,
is contaminated by the 
auto-correlation of the source signature (Figure~\ref{fig:wavobs}a).
Therefore, 
we use FBD to first extract the interferometric Green's functions by FIBD, plotted in 
the Figure~\ref{fig:reflec_main_marm}e, and then recover the 
Green's functions, plotted in the Figure~\ref{fig:reflec_main_marm}g, using FPR.
Notice that the FBD estimated Green's functions 
clearly depict 
the direct and the scattered arrivals, confirming that 
our method
doesn't suffer from the complexities in the subsurface models.

\section{Conclusions}
\label{sec:conc}
Focused blind deconvolution (FBD)
solves a series of 
two optimization problems 
in order to perform
multichannel blind deconvolution (BD), 
where both the unknown impulse responses and the unknown source signature are estimated 
given the channel outputs.
It is designed for a BD problem where the impulse responses are 
supposed to be sparse, front-loaded and shorter in duration compared to the channel outputs; 
as in the case of seismic inversion with a noisy source.
The optimization problems use focusing constraints to 
resolve the indeterminacy inherent to the traditional BD.
The first problem 
considers fitting the interferometric channel outputs and 
focuses the energy of the impulse-response auto-correlations
at the zero lag to estimate the interferometric impulse responses
and the source auto-correlation.
The second problem completes FBD by 
fitting the estimated interferometric impulse responses, 
while 
focusing the energy 
of the most front-loaded channel 
at the zero time.
FBD doesn't require any support constraints on the unknowns.
We have demonstrated the benefits of FBD 
using seismic 
experiments.  
%
%
%

	%

\section{Acknowledgements}

The material is based upon work
assisted by a grant from Equinor. 
Any opinions, findings, and conclusions or recommendations
expressed in this material 
are those of the authors
and do not necessarily 
reflect the views of Equinor.
The authors thank 
Ali Ahmed, Antoine Paris, Dmitry Batenkov and Matt Li from MIT for helpful discussions, and Ioan Alexandru Merciu
from Equinor for
his informative review and commentary of a draft version. 
LD is also supported by AFOSR grant FA9550-17-1-0316, and NSF grant DMS-1255203.

\onecolumn
\bibliographystyle{seg}
\bibliography{sample,../ica_seg/paper/ica_ref}

\appendix
\section{Appendix}

In this 
appendix, 
we present a simple justification of the ability of a focusing functional on the autocorrelation to select for sparsity, in a setting where $\ell_1$ minimization is unable to do so. This setting is the special case of a vector with \emph{nonnegative} entries, made less sparse by convolution with another vector with nonnegative entries as well. This scenario is not fully representative of the more general formulation assumed in this paper, where cancelations may occur because of alternating signs. It seems necessary, however, to make an assumption of no cancelation (like positivity) in order to obtain the type of comparison result that we show in this section. 

Consider two infinite sequences $f_i$ and $\phi_j$, for $i, j \in \Z$ (the set of integers), with sufficient decay so that all the expressions below make sense, and all the sum swaps are valid. Assume that $f_i \geq 0$ and $\phi_i \geq 0$ for all $i \in \Z$, not identically zero. Let
\[
g_j = (f * \phi)_j = \sum_{i \in \Z} f_i \phi_{j-i},
\]
which obviously also obeys $g_i \geq 0$ for all $i \in Z$. Assume the normalization condition $\sum_{i \in \Z} \phi_i = 1$.

Now consider the autocorrelations
\[
F_j = (f \otimes f)_j = \sum_i f_i f_{j+i}, \qquad G_j = (g \otimes g)_j = \sum_i g_i g_{j+i},
\]
and a specific choice of focusing functional,
\[
J_F = \sum_{j \in Z} j^2 F_j, \qquad J_G = \sum_{j\in \Z} j^2 G_j.
\]

\begin{proposition}
\[
J_G \geq J_F,
\]
with equality if and only if $\phi_i$ is the Kronecker $\delta_{i0}$
.
\end{proposition}
\begin{proof}
All sums run over $\Z$. Start by observing
\[
J_F = \sum_j \sum_k K_{jk} f_j f_k, \qquad K_{jk} = (j-k)^2,
\]
and
\[
J_G = \sum_j \sum_k L_{jk} f_j f_k, \qquad L_{jk} = \sum_m \sum_n \left( (j-k) - (m-n) \right)^2 \phi_m \phi_n.
\]
For any particular value $m - n = a$, we have
\begin{align*}
\sum_j \sum_k \left( (j-k) - a \right)^2 f_j f_k &= a^2 \sum_j \sum_k f_j f_k + \sum_j \sum_k (j-k)^2 f_j f_k \\
&\geq J_F,
\end{align*}
(the term linear in $j-k$ drops because $j-k$ is antisymmetric in $j$ and $k$, while $f_j f_k$ is symmetric), with equality if and only if $a = 0$.

Now $J_G$ is a convex combination of such contributions:
\begin{align*}
\sum_m \sum_n \left[ \sum_j \sum_k \left( (j-k) - (m-n) \right)^2 f_j f_k \right] \phi_m \phi_n & \geq \sum_m \sum_n \left[ J_F \right]  \phi_m \phi_n \\
&= J_F
\end{align*}
with equality if and only if the cartesian product supp $\phi \; \times $ supp $\phi$ contains only the diagonal $m = n$. This latter scenario only arises when supp $\phi = \{ 0 \}$, which is only compatible with $\sum_i \phi_i = 1$ when $\phi_i = \delta_{i0}$.
\end{proof}

In contrast, notice that $\sum_i f_i = \sum_i g_i$, hence $f$ and $g$ cannot be discriminated with the $\ell_1$ norm. The $\ell_1$ norm is unable to measure the extent to which the support of $f$ was ``spread" by convolution with $\phi$, when $\sum \phi_i = 1$, and when all the functions are nonnegative.

The continuous counterpart of this result, for nonnegative functions $f(t)$ and $g(t) = \int f(s) \phi(t-s) ds$, with nonnegative $\phi$ such that $\int \phi(t) dt = 1$ in the sense of measures, involves the autocorrelations
\[
F(t) = (f \otimes f)(t) = \int f(s) f(s+t) ds, \qquad G(t) = (g \otimes g)(t),
\]
and focusing functionals
\[
J_F = \int t^2 F(t) dt, \qquad J_G = \int t^2 G(t) dt.
\]
Then, $J_G \geq J_F$, with equality if and only if $\phi(t) = \delta(t)$, the Dirac delta.

\end{document}